\documentclass[11pt,letterpaper]{article}
\usepackage[T1]{fontenc}
\usepackage[sc]{mathpazo}
\usepackage{array}
\usepackage{microtype}
\usepackage{amsmath}
\usepackage[dvipsnames]{xcolor}
\usepackage{amssymb}
\usepackage{amsthm}
\usepackage{bm,float}
\usepackage{dsfont}
\usepackage{authblk}
\usepackage{fullpage,caption,wrapfig}
\usepackage{comment}
\usepackage{mathtools}
\usepackage[shortlabels,inline]{enumitem}
\usepackage{bbm}
\usepackage{booktabs}
\usepackage{hyperref}
\definecolor{pku-red}{RGB}{139,0,18}
\hypersetup{
    colorlinks=true,
    linkcolor=pku-red,     
    urlcolor=violet,
    citecolor=BlueViolet,
    pdffitwindow=true,
}\usepackage[capitalize, nameinlink]{cleveref}
\usepackage{nag,tikz,pgfplots}
\pgfplotsset{compat=newest}

\usetikzlibrary{arrows.meta}%
\usetikzlibrary{arrows}
\usepgfplotslibrary{colormaps,fillbetween}
\usetikzlibrary{arrows.meta, positioning, decorations.markings}

\usepackage{algorithm,algpseudocode}
\usepackage{multicol}
\usepackage[scaled]{helvet} 
\usepackage{thmtools} 

\usepackage{crossreftools}
\numberwithin{equation}{section}

\theoremstyle{plain}
\newtheorem{theorem}{Theorem}[section]
\newtheorem{lemma}[theorem]{Lemma}
\newtheorem{corollary}[theorem]{Corollary}
\newtheorem{proposition}[theorem]{Proposition}

\theoremstyle{definition}
\newtheorem{definition}[theorem]{Definition}

\theoremstyle{remark}
\newtheorem{remark}[theorem]{Remark}

\newcommand{\R}{\mathbb R}

\newcommand{\N}{\mathbb N}

\newcommand{\G}{\mathcal G}
\newcommand{\uf}{\mathfrak{uf}}
\DeclareMathOperator*{\fld}{fld}

\newcommand{\T}{\mathsf T}

\newcommand{\I}{\mathcal I}

\DeclareMathOperator*{\supp}{supp}

\DeclareMathOperator*{\lcm}{lcm}

\newcommand{\s}[1]{{(#1)}}

\newcommand{\norm}[1]{\left\lVert{#1}\right\rVert}

\title{Will AI Trade? A Computational Inversion of the No-Trade Theorem}

\author[1]{Hanyu Li\thanks{lhydave@pku.edu.cn}}
\author[1]{Xiaotie Deng\thanks{xiaotie@pku.edu.cn}}
\affil[1]{CFCS, School of Computer Science, Peking University}
\date{}

\begin{document}
\maketitle 
\begin{abstract}
Classic no-trade theorems attribute trade to heterogeneous beliefs. We re-examine this conclusion for AI agents, asking if trade can arise from computational limitations, under common beliefs. We model agents' bounded computational rationality within an unfolding game framework, where computational power determines the complexity of its strategy. Our central finding inverts the classic paradigm: a stable no-trade outcome (Nash equilibrium) is reached only when ``almost rational'' agents have slightly different computational power. Paradoxically, when agents possess identical power, they may fail to converge to equilibrium, resulting in persistent strategic adjustments that constitute a form of trade. This instability is exacerbated if agents can strategically under-utilize their computational resources, which eliminates any chance of equilibrium in Matching Pennies scenarios. Our results suggest that the inherent computational limitations of AI agents can lead to situations where equilibrium is not reached, creating a more lively and unpredictable trade environment than traditional models would predict.    
\end{abstract}
\section{Introduction}
A foundational economic question is what drives trade. The classic ``no-trade'' theorem~\cite{milgromInformationTradeCommon1982}, stemming from Aumann's Agreement Theorem~\cite{aumannAgreeingDisagree1976}, posits that rational agents with common priors should not trade based on private information, as any agreed-upon price would reveal enough to nullify the trade's appeal for one party. This information-theoretic view suggests trade arises from heterogeneous beliefs~\cite{morrisTradeHeterogeneousPrior1994}.

However, the rise of sophisticated AI agents participating in economic settings prompts a re-examination of this classic result~\cite{openaiGPT4TechnicalReport2024,geminiteamgoogleGemini25Pushing2025,guoDeepSeekR1IncentivizesReasoning2025}. While AI agents can process vast amounts of information, they are ultimately bound by finite computational resources. This introduces a different dimension of rationality --- \emph{computational rationality} --- which is distinct from the informational assumptions underlying the no-trade theorems~\cite{halpernAlgorithmicRationalityGame2015,neymanFinitelyRepeatedGames1998,rubinsteinFiniteAutomataPlay1986}. 

This paper explores the ``Will AI trade?'' question through the lens of computational boundedness. We ask: if we assume agents are informationally aligned (i.e., they share common priors), can limitations in their computational power create conditions for trade?

We model this multi-agent environment using an ``unfolding game,'' a framework that reinterprets a classic single-shot normal-form game as an infinite interaction, thereby creating a setting to analyze computational effort. In this model, strategies are represented by infinite action sequences and payoffs by their long-term average, allowing us to define an agent's ``rationality level'' by the complexity of the sequence it can produce.

Our initial analysis assumes agents always exert their maximum computational power, an ``arms race'' mentality. Our focus is on ``almost rational'' agents --- those with immense and growing computational power, mirroring the current trajectory of AI development. We are interested in the conditions under which these agents can reach a stable outcome, defined as a Nash equilibrium, where no agent has an incentive to unilaterally change its strategy --- a ``no-trade'' state as is demonstrated in the no-trade theorem~\cite{milgromInformationTradeCommon1982}.

Our central finding inverts the classic no-trade paradigm. We show that:
\begin{enumerate}
    \item If almost rational agents have \textbf{slightly different} computational power, they reach a stable Nash equilibrium (a ``no-trade'' state).
    \item If these agents have \textbf{perfectly identical} computational power, they could fail to reach any stable equilibrium in some games, leading to persistent strategic adjustments akin to trade.
\end{enumerate}

This result suggests that from a computational perspective, it is \textbf{homogeneity}, not heterogeneity, that can lead to the absence of equilibrium and thus drive trade.

We then relax this ``arms race'' assumption of maximum computational power and consider the case where agents can strategically ``pretend'' to be computationally weaker than they are. As an example, we show that in the Matching Pennies game, this flexibility completely undermines the potential for any Nash equilibria, i.e., stable outcomes. If an agent can choose to under-utilize its power, the ability to predict its behavior breaks down, making it impossible to guarantee convergence to a Nash equilibrium, and thus trade happens.

Our findings offer a new, computation-driven model for understanding when trade might occur due to the absence of equilibrium, contributing to AI economics. We demonstrate that even with informational alignment, computational rationality's architecture fundamentally alters classic economic predictions. This highlights that for large-scale agent interactions, such as on the Internet, computational constraints are as crucial as informational ones. Paradoxically, the pursuit of greater, more uniform AI computational power might not always lead to stable no-trade outcomes but instead create conditions where equilibrium is not reached, fostering new forms of ongoing interaction that resemble trade.
\section{Unfolding Games: A Model for Computationally Bounded Agents}

Classic no-trade theorems are rooted in Bayesian games, where trade is driven by informational heterogeneity --- differences in priors, or ``types''~\cite{harsanyiGamesIncompleteInformation1967,harsanyiGamesIncompleteInformation1968,harsanyiGamesIncompleteInformation1968a}. To isolate the role of \emph{computational} limits, we remove informational asymmetry from the model. We begin with a standard two-player, finite normal-form game, which can be seen as a Bayesian game with a single state of the world and both players have the same unique type. Thus, a common prior belief structure is inherent. In this setting, where players have complete and identical information about the game itself, traditional theory predicts no incentive for trade between two players. Our goal is to show how computational constraints can lead to trade.

\subsection{From Single Decisions to Infinite Sequences}

We interpret a single-shot normal-form game as an infinitely repeated interaction. In this view, a player does not choose a single action or a probability distribution over actions (mixed strategy), but rather commits to an infinite sequence of pure actions. The player's mixed strategy is then interpreted by the long-run frequency of these actions (the ``frequentist interpretation of mixed strategies~\cite{hajekInterpretationsProbability2019,billotProbabilitiesSimilarityWeightedFrequencies2005}''). We formalize this using the concept of an ``unfolding game.''

\begin{definition}[Unfolding game]
Let $\G = (\{1, 2\}, \{S_1, S_2\}, \{u_1, u_2\})$ be a two-player normal-form game, where $S_i$ is player $i$'s finite set of pure actions and $u_i: S_1 \times S_2 \to \R$ is their payoff function.

The \textbf{unfolding} of $\G$ is a game $\G^\uf$ where each player $i$'s strategy space is $S_i^\infty$, the set of all infinite sequences of actions $s_i = (s_i(1), s_i(2), \dots)$. A strategy profile is a pair of such sequences $s = (s_1, s_2)$. The payoff to player $i$ for a strategy profile $s$ is the long-run average payoff:
\[
u_i^\uf(s) = \liminf_{n\to\infty} \frac{1}{n} \sum_{j=1}^n u_i(s_1(j), s_2(j)).
\footnote{The limit may not always exist. However, since the payoff function $u_i$ takes only a finite number of values, the lower limit always exists. We interpret this as the minimum payoff the player is guaranteed to secure.}
\]
Here, $s_k(j)$ is the action of player $k$ in round $j$.
\end{definition}

Thus, in the unfolding game $\G^\uf$, players choose \textit{at the beginning of the game} a strategy in the form of an infinite action sequence, which they then follow throughout the game.

Conversely, we can interpret any such infinite sequence as a mixed strategy in the original game $\G$ via a \textbf{folding} operation. 

\begin{definition}[Folding]
For a given strategy sequence $s_i$, its corresponding \textbf{folding} mixed strategy $\mu_i$ is the frequency vector of its actions, provided the limit exists:
\[
\mu_i(a) = \lim_{n \to \infty} \frac{1}{n}|\{j \in [1, n]: s_i(j) = a\}|, \quad \forall a \in S_i.
\]
Here, $\mu_i(a)$ is the probability that player $i$ plays action $a$ in the mixed strategy. A strategy profile $s = (s_1, s_2)$ in the unfolding game folds to a mixed strategy profile $\mu = (\mu_1, \mu_2)$ in the original game. 
\end{definition}
This relationship is summarized below.

\begin{table}[h]
\centering
\caption{Interpretation of Mixed Strategies and Payoffs}
\label{tab:interpretation}
\begin{tabular}{cc}
\toprule
\textbf{Unfolding Game $\G^\uf$} & \textbf{Original Game $\G$} \\ \midrule
Action Sequence $s_i$ & Mixed Strategy $\mu_i$ \\
Average Payoff $u_i^\uf(s)$ & Expected Payoff $u_i(\mu)$ \\ \bottomrule
\end{tabular}
\end{table}

The unfolding game framework is introduced as a technical device to rigorously analyze computational effort within the static, single-shot normal-form game. Unlike repeated games, it does not model agents adapting to observed histories, but rather provides a way to formalize the complexity of strategies agents can commit to at the outset.

\subsection{Defining Rationality as Computational Power}\label{sec:unfolding-game-rationality-definition}

We model players as computationally bounded agents. An agent's ``rationality'' is not a measure of their prior or information, but of their raw computational capacity to execute a complex strategy. We model this by restricting players' strategy spaces to \textbf{eventually repeated sequences}. Such a sequence takes the form $s = xyyy\dots$, where $x$ is a finite prefix and $y$ is a finite sequence that repeats infinitely. Actually, eventually repeated sequences are the only ones that can be produced by pseudo-random Turing machines with finite memory with game structure as input, the most generic model for AI agents. See \Cref{sec:eventually-repeated} for a more detailed discussion.

While traditional machine learning approaches are adaptive or update based on historical interactions, our model is inspired by a different but equally important class of AI agents, particularly Large Language Models (LLMs)~\cite{openaiGPT4TechnicalReport2024,geminiteamgoogleGemini25Pushing2025,guoDeepSeekR1IncentivizesReasoning2025}. During inference, an LLM has a fixed set of weights and generates a sequence of actions without updating its core parameters. Our model of a Turing machine with finite memory committing to a sequence is a powerful abstraction for this type of agent. It is a general model of computation that encompasses a vast range of algorithmic agents in their deployed, non-training phase.

In the long-run average payoff, any finite prefix becomes negligible. The complexity of a strategy is therefore reflected by the length of its repeating part, or its period. Our assumption is that executing a strategy with a longer, more intricate period requires greater computational resources (e.g., memory, processing power). This leads to our central definition of computational rationality.

\begin{definition}[Rationality level]
A player $i$ has \textbf{rationality level} $\tau_i \in \N$ if their strategy space, denoted by $A_i^\uf$, is restricted to eventually repeated sequences with a period of length $\tau_i$. That is, the player can only play strategies of the form $s_i = x y^{\infty}=xyyy\dots$, where $x$ is a finite prefix and $y$ is a sequence of length $\tau_i$ that repeats infinitely.
\end{definition}

\begin{remark}
    Such measure of rationality levels uses infinite sequences, which can be found in the literature on computationally bounded rationality in other scenarios~\cite{neymanFinitelyRepeatedGames1998,rubinsteinFiniteAutomataPlay1986,gutierrezImperfectInformationReactive2018,gilboaInfiniteHistoriesSteady1994,kalaiBoundedRationalityStrategic1990,gilboaBoundedUnboundedRationality1989}.
\end{remark}

We make a critical behavioral assumption that mirrors a technological ``arms race'': players are compelled to use their full computational power. An agent with rationality level $\tau_i$ \emph{must} play a strategy with period $\tau_i$; they cannot choose to play a simpler strategy of a shorter period. While this assumption would be ideal and ignore the complex reality of strategic decision-making, it allows us to cleanly investigate the consequences of agents' peak computational abilities. In \Cref{sec:strategic-under-utilization}, we will relax this assumption to explore the strategic implications of agents being able to \emph{pretend} to be computationally weaker than they are.

\section{Almost-Rational Agents and the Conditions for Trade}\label{sec:trade-condition}

Having established the unfolding game as a model for computational rationality, we now investigate the conditions under which computationally bounded agents would trade. The classical no-trade theorems conclude that in a Bayesian game, rational agents will not trade at Bayesian Nash equilibria (BNE)~\cite{harsanyiGamesIncompleteInformation1968,aumannAgreeingDisagree1976,milgromInformationTradeCommon1982}. In a BNE, agents possess private information (their ``type'') and use Bayesian updating to interpret others' actions. With common prior beliefs, any offer to trade reveals information about the offerer's type, leading rational agents to infer that the trade would not be beneficial, thus preventing agents from deviating from their current strategies, or in other words, no trade occurs.

A Bayesian game with a common prior and different states of the world is formally equivalent to a single, larger normal-form game. In this equivalent game:
\begin{itemize}
    \item A player's pure strategy is not just a single action, but a complete contingency plan (a function) that maps every possible state of the world to an action.
    \item The payoffs are the expected payoffs from the original Bayesian game, averaged over the common prior distribution of the world states.
\end{itemize}
Under this equivalence, a Bayesian Nash Equilibrium (BNE) in the original game corresponds precisely to a Nash Equilibrium (NE) in this larger normal-form game.

Thus, when considering normal-form games --- a special case of Bayesian games with a single state of the world --- Bayesian Nash equilibria (BNEs) degenerate to Nash equilibria (NEs)~\cite{nashNonCooperativeGames1951}. We consider NEs, where no player can unilaterally improve their payoff by changing strategies, as stable ``no-trade'' outcomes. Conversely, a failure to reach an NE implies ongoing strategic adjustments, a dynamic we interpret as trade driven by computational, rather than informational, pressures. An illustration of this correspondence is shown below.
\[
\text{Reach a Nash equilibrium} \iff \text{Trade does not occur}
\]

\subsection{Almost Identical Rationality}

We are particularly interested in ``almost rational'' agents --- those whose computational power is immense and continually growing, mirroring the trajectory of modern AI development. We formalize this by defining a player's rationality level not as a fixed number, but as a sequence of natural numbers tending to infinity.

\begin{definition}[Almost rationality]
    An agent $i$ is \textbf{almost rational} if their rationality level is a sequence $\{\tau_i(n)\}_{n=1}^\infty$ where $\lim_{n\to\infty} \tau_i(n) = +\infty$. At each stage $n$, the agent is bound to play a strategy with a period of exactly $\tau_i(n)$.
\end{definition}

To model the scenario of competing AIs developed with similar architectures and resources, we introduce a precise definition for when their computational powers are almost the same.

\begin{definition}[Almost identical rationality]
 Two almost rational agents with rationality levels $\{\tau_1(n)\}_{n=1}^\infty$ and $\{\tau_2(n)\}_{n=1}^\infty$ are \textbf{almost identical} if the relative difference in their computational power vanishes at infinity. That is,
 \[\lim_{n\to\infty}\frac{|\tau_1(n)-\tau_2(n)|}{\min\{\tau_1(n),\tau_2(n)\}}=0.\]
\end{definition}

This definition captures the intuition that while two powerful AIs might have different processing speeds or memory, these differences become insignificant relative to their massive total capacities. Crucially, this definition includes two opposite (while not complementary) cases that will be central to our analysis:
\begin{itemize}
    \item \textbf{Almost identical but always slightly different:} The agents have rationality levels that differ by a small, bounded amount for all but finitely many stages. For example, $\tau_1(n) = n$ and $\tau_2(n) = n + 1$ for all $n > N$.
    \item \textbf{Perfectly identical:} The agents have rationality levels that are exactly the same for all sufficiently large $n$. For example, $\tau_1(n) = n$ and $\tau_2(n) = n$ for all $n > N$.
\end{itemize}

\subsection{Nash Equilibria for Almost Rational Agents}

To analyze the equilibrium outcomes for almost rational agents in unfolding games, we need an appropriate notion of Nash equilibria that accommodates their evolving computational capabilities. As discussed earlier, a Nash equilibrium represents a stable outcome where no agent has an incentive to unilaterally deviate, thus corresponding to a ``no-trade'' situation. 

Standard Nash equilibrium applies to a game with fixed strategy spaces. However, in our model, as agents become ``more rational'' (i.e., as $n$ increases), their rationality levels $\tau_1(n)$ and $\tau_2(n)$ increase, and consequently, their available strategy spaces $A_1^\uf(n)$ and $A_2^\uf(n)$ in the unfolding game $\G^\uf(n)$ also change. We are interested in whether these agents can reach a stable outcome in the limit as their rationality tends to infinity.

To capture this, we consider a sequence of strategy profiles, $\{s^{(n)}\}_{n=1}^\infty$, where each $s^{(n)}$ is played in the unfolding game $\G^\uf(n)$ corresponding to the rationality levels $(\tau_1(n), \tau_2(n))$. We want to define an equilibrium concept that captures the idea that as $n \to \infty$ (and thus as agents become increasingly rational), these strategies $s^{(n)}$ become increasingly stable, meaning the incentive for any player to unilaterally deviate approaches zero.

We thus first define the notion of an $\epsilon$-approximate Nash equilibrium for fixed rationality levels as follows.

\begin{definition}[$\epsilon$-approximate Nash equilibrium for bounded rationality]
Consider a two-player normal-form game $\G = (\{1, 2\}, \{S_1, S_2\}, \{u_1, u_2\})$, where $S_i$ is player $i$'s finite set of pure actions and $u_i: S_1 \times S_2 \to \R$ is their payoff function. Let $\G^\uf$ be the unfolding game of $\G$ where player $i$ has rationality level $\tau_i \in \N$. Thus, player $i$'s strategy space in $\G^\uf$ is $A_i^\uf$, the set of all eventually repeated sequences with period $\tau_i$.

A strategy profile $s = (s_1, s_2)$ in $\G^\uf$ is called an \textbf{$\epsilon$-approximate Nash equilibrium} ($\epsilon$-NE) if for each player $i \in \{1, 2\}$,
\[
\max_{s_i' \in A_i^\uf} u_i^\uf(s_i', s_j) - u_i^\uf(s) \leq \epsilon,
\]
where $s_j$ is the strategy of the other player. Here, $u_i^\uf(s)$ is the long-run average payoff of player $i$ when playing strategy profile $s$ in the unfolding game $\G^\uf$.

\end{definition}

Using this definition, we can now define a Nash equilibrium for almost rational agents in the unfolding game.

\begin{definition}[Nash equilibrium for almost rational agents]\label{def:NE-for-almost-rational}
Consider a two-player normal-form game $\G = (\{1, 2\}, \{S_1, S_2\}, \{u_1, u_2\})$. Let $\{\tau_1(n)\}_{n=1}^\infty$ and $\{\tau_2(n)\}_{n=1}^\infty$ be the rationality levels of the two almost rational agents. For each $n$, let $\G^\uf(n)$ be the unfolding game where player $i$ has rationality level $\tau_i(n)$ and strategy space $A_i^\uf(n)$ (the set of all eventually repeated sequences with period $\tau_i(n)$).

A sequence of strategy profiles $\{s^{(n)}\}_{n=1}^\infty$, where $s^{(n)} = (s_1^{(n)}, s_2^{(n)})$ with $s_i^{(n)} \in A_i^\uf(n)$, is a \textbf{Nash equilibrium} if for each $n$, $s^{(n)}$ is an $\epsilon_n$-NE of $\G^\uf(n)$ with $\lim_{n\to\infty} \epsilon_n=0$.
\end{definition}

In essence, a sequence of strategies constitutes a Nash equilibrium for almost rational agents if, as their computational power grows indefinitely, the maximum possible gain any player could achieve by unilaterally changing their strategy in the corresponding unfolding game diminishes to zero. This captures the notion of achieving a stable, equilibrium-like outcome in the limit of growing bounded rationality.

\subsection{Main Result: Identical Rationality as a Driver of Trade}

Our main theorem delineates the conditions under which almost rational agents achieve stable ``no-trade'' outcomes (Nash equilibria). It reveals that the decisive factor is whether their rationality levels are almost identical but always slightly different.

\begin{theorem}\label{thm:main-result}
Consider two players with almost rationality levels $\{\tau_1(n)\}_{n=1}^\infty$ and $\{\tau_2(n)\}_{n=1}^\infty$ that are almost identical. The following statements are equivalent:
\begin{enumerate}
    \item The rationality levels of both players are almost identical but always slightly different, i.e., there exists an $N\in\N$, for every $n>N$, $\tau_1(n)\neq\tau_2(n)$.
    \item For every two-player finite normal-form game $\G$, the two players can reach a Nash equilibrium in the sense of \Cref{def:NE-for-almost-rational}.
\end{enumerate}
\end{theorem}

This theorem highlights a critical distinction for the conditions under which trade may occur. Interpreting the theorem, ``trade'' is possible (i.e., a Nash equilibrium may not be reached in some games) if agents are perfectly identical in their rationality levels, i.e., $\tau_1(n) = \tau_2(n)$ for all sufficiently large $n$.\footnote{Actually, the precise condition for the possibility of trade is that $\tau_1(n) = \tau_2(n)$ for infinitely many $n$.}

The core implication for the ``almost identical'' AIs central to our model is a direct contrast:
\begin{itemize}
    \item \textbf{Slight Heterogeneity $\Rightarrow$ No Trade:} If agents are almost identical and their rationality levels remain asymptotically different ($\tau_1(n) \neq \tau_2(n)$ for large $n$), a Nash equilibrium is always reachable. This ensures a stable, ``no-trade'' outcome.
    \item \textbf{Perfect Homogeneity $\Rightarrow$ Potential for Trade:} If agents are almost identical but their rationality levels become asymptotically the same ($\tau_1(n) = \tau_2(n)$ for large $n$), failure to reach a Nash equilibrium (and thus, ``trade'' as persistent strategic adjustment) is possible in some games.
\end{itemize}
This presents a computational inversion of classic no-trade theorems. Assuming agents always deploy their maximal computational power, it is \textit{perfect homogeneity} in computational capabilities, rather than heterogeneity, that can lead to the absence of equilibrium and thus drive trade.

\subsection{Correspondence with Equilibria in the Original Game}

One might argue that the unfolding game, $\G^\uf$, is a distinct scenario from the original normal-form game, $\G$, potentially limiting the applicability of our findings to classical game theory and no-trade theorems. However, we argue that our framework maintains a crucial correspondence. When the conditions for achieving a Nash equilibrium in the unfolding game are satisfied (as stated in \Cref{thm:main-result}), this stability translates meaningfully back to the original game. 

Specifically, any Nash equilibrium of the original normal-form game can be realized as the limit of a Nash equilibrium sequence in \Cref{def:NE-for-almost-rational}, both in terms of strategy profiles (via folding) and payoffs. The following theorem formalizes this connection.

\begin{theorem}[Correspondence of equilibria]\label{thm:correspondence}
Consider a two-player finite normal-form game $\G = (\{1, 2\}, \{S_1, S_2\}, \{u_1, u_2\})$. Suppose the rationality levels $\{\tau_1(n)\}_{n=1}^\infty$ and $\{\tau_2(n)\}_{n=1}^\infty$ of the two almost rational agents are almost identical but always slightly different (i.e., there exists an $N_0 \in \N$ such that for every $n > N_0$, $\tau_1(n) \neq \tau_2(n)$).

Then, for any Nash equilibrium $\sigma^* = (\sigma_1^*, \sigma_2^*)$ of the original game $\G$, there exists a sequence of strategy profiles $\{s^{*(n)}\}_{n=1}^\infty$, where $s^{*(n)} = (s_1^{*(n)}, s_2^{*(n)})$ with $s_i^{*(n)} \in A_i^\uf(n)$, such that:
\begin{enumerate}
    \item $\{s^{*(n)}\}_{n=1}^\infty$ is a Nash equilibrium for the almost rational agents in the sense of \Cref{def:NE-for-almost-rational}.
    \item For each player $i \in \{1,2\}$, the folding of $s_i^{*(n)}$, denoted $\mu_i^{(n)}$, converges to $\sigma_i^*$ as $n \to \infty$. That is, for every action $a \in S_i$, $\lim_{n\to\infty} \mu_i^{(n)}(a) = \sigma_i^*(a)$.
    \item For each player $i \in \{1,2\}$, the long-run average payoffs in the unfolding game converge to the expected payoffs in the original game's Nash equilibrium. That is, $\lim_{n\to\infty} u_i^\uf(s^{*(n)}) = u_i(\sigma^*)$.
\end{enumerate}
\end{theorem}

\Cref{thm:correspondence} ensures that any Nash equilibrium of the original static game can be realized as the limit of a stable sequence of equilibria in the unfolding game, under the ``almost identical but different'' condition. This serves as a sanity check that our unfolding framework faithfully captures the equilibrium structure of the original game.

\subsection{Proof Idea and Example}
\label{sec:proof-idea}

The proofs for our main findings (\Cref{thm:main-result} and \Cref{thm:correspondence}) are deeply connected and are proven together in \Cref{sec:counterpoint}. The core idea is that we can translate the equilibrium analysis from the complex unfolding game back to the original, simpler game. We do this by ``folding'' the players' repeating action sequences into familiar mixed strategies. To make this abstract idea concrete, we will use the classic game of Matching Pennies as an illustration.

In Matching Pennies, two players each choose to show either heads ($H$) or tails ($T$). Player 1, the row player, wins if their choices match, while player 2, the column player, wins if they differ. The payoff matrices are as follows:
\begin{equation}\label{eq:matching-pennies-RC}
    u_1:\begin{array}{c|cc}
            &H&T  \\\hline
         H& 1&0\\
         T&0&1\\
    \end{array}\qquad
    u_2:\begin{array}{c|cc}
            &H&T  \\\hline
         H& 0&1\\
         T&1&0\\
    \end{array}
\end{equation}

We will now explore two contrasting scenarios that reveal the intuition behind our main theorem.

\subsubsection*{Case 1: The ``No-Trade'' Harmony (Slightly Different Rationality)}

This case illustrates why slightly different computational power leads to a stable Nash equilibrium (a ``no-trade'' outcome).

Let's assume player 1 has a rationality level of $\tau_1(n) = n$ and player 2 has $\tau_2(n) = n+1$. For $n=2$, player 1 has a period of 2 and player 2 has a period of 3. Suppose player 1 plays the repeating sequence $(HT)^\infty$ and player 2 plays $(HTT)^\infty$.

\paragraph{Aligning the Rhythms} To analyze this interaction, we must find a common time horizon where both players' patterns align. The least common multiple of their periods (2 and 3) is 6. So, we expand their strategies to a 6-round cycle:
\begin{itemize}
    \item Player 1's strategy $(HT)^\infty$ becomes $(HTHTHT)^\infty$.
    \item Player 2's strategy $(HTT)^\infty$ becomes $(HTTHTT)^\infty$.
\end{itemize}
The play-by-play over this 6-round piece looks like this:
\[
\begin{array}{c|ccccccc|c}
\text{Round} & 1 & 2 & 3 & 4 & 5 & 6 & \dots & \text{Average Payoff} \\\hline
s_1 & H & T & H & T & H & T & \dots & \\
s_2 & H & T & T & H & T & T & \dots & \\\hline
u_1 & 1 & 1 & 0 & 0 & 0 & 1 & \dots & 3/6 = 1/2 \text{ } \\
u_2 & 0 & 0 & 1 & 1 & 1 & 0 & \dots & 3/6 = 1/2 \text{ } \\
\end{array}
\]

\paragraph{The Counterpoint Insight} The key observation lies in how the actions meet over this 6-round cycle. Because the periods 2 and 3 are coprime, every action in player 1's cycle encounters \emph{every} action in player 2's cycle \emph{exactly once}. For example, player 1's first action ($H$) plays against player 2's $H$, $T$, and $T$ over the cycle.

This perfect mixing means that from player 1's perspective, player 2's strategy effectively \emph{behaves like} a fixed mixed strategy. We can calculate this ``folding'' by looking at the frequency of actions in the 6-round cycle:
\begin{itemize}
    \item Player 1's folding: $(H: 3/6, T: 3/6) \rightarrow \sigma_1=(H:1/2, T:1/2)$.
    \item Player 2's folding: $(H: 2/6, T: 4/6) \rightarrow \sigma_2=(H:1/3, T:2/3)$.
\end{itemize}

\paragraph{Calculating the Profit from Deviation} Player 1 can now calculate their best response against player 2's folding mixed strategy $\sigma_2=(H:1/3, T:2/3)$. The best response to this is to always play $T$. If player 1 deviates to a strategy of $(TT)^\infty$, which behaves like the pure strategy $\sigma_1'=(H:0, T:1)$, their new payoff would be $2/3$.

Here we see the crucial trade-off: this limited potential gain is a direct consequence of the behavioral constraints imposed by our model. Having committed to using their full computational power, each agent must adhere to a fixed, repeating strategic pattern. This inflexibility --- the ``price'' of maximizing computational output --- confines them to a small set of strategic behaviors and prevents large, opportunistic gains from deviation.

The maximum extra payoff player 1 can gain is the difference between the new and old payoffs: $2/3 - 1/2 = 1/6$. Since player 2 is already playing her best response, this small, non-zero gain of $1/6$ means the strategies form a $1/6$-NE. 

More generally, this $1/6$ gap arises because the ``folded'' outcome of the players' strategies is not a perfect NE. As the rationality level $n$ increases, players can construct repeating sequences whose folding more precisely approximates a true NE of the original game. As $n \to \infty$, this slight imperfection, and the potential payoff gain it allows, diminishes to zero, leading to a NE outcome in the limit.

\subsubsection*{Case 2: The ``Trade'' Dissonance (Identical Rationality)}

This case illustrates why perfectly identical computational power can prevent a stable equilibrium, thus creating conditions for ``trade''. Suppose both players have the same rationality level, $\tau_1(n) = \tau_2(n) = n$.

\paragraph{The Problem of Perfect Alignment} When both players have the same period, their strategies align perfectly round after round. There is no ``mixing'' effect from misaligned cycles. This leads to a constant feeling of regret for one of the players.

\paragraph{The Guaranteed Deviation Payoff} Since Matching Pennies is a zero-sum game, one player (let's say player 1) must be receiving an average payoff of no more than $1/2$. In any given round, player 1 can see what player 2 is playing and think, ``If only I had chosen differently, I would have won''.

More precisely, player 1 can always deviate to a strategy that plays the best response to player 2's action in \emph{every single round}. Such a deviation strategy would yield a payoff of 1 in every round, guaranteeing an average payoff of 1. Therefore, by deviating, player 1 can always gain an extra payoff of at least $1 - 1/2 = 1/2$.

Because the potential gain from deviation ($\epsilon_n$) is always at least $1/2$ and never approaches zero, the players can never reach a Nash equilibrium. This persistent instability and strategic adjustment is what we interpret as trade.

\paragraph{The Counterpoint Technique}
The general proof formalizes this intuition. It requires number theory tools to handle the alignment of sequences with different periods. We call this the \textit{counterpoint technique}, borrowing a term from music for the art of harmoniously combining different melodic lines (strategies). The detailed proof is presented in \Cref{sec:counterpoint}.
\section{Strategic Under-utilization Leading to Trade}
\label{sec:strategic-under-utilization}

In the preceding sections, we assumed that agents, akin to participants in a computational ``arms race,'' are compelled to utilize their maximum available computational power. This meant an agent with rationality level $\tau_i$ was restricted to playing strategies with a period of exactly $\tau_i$. We now relax this assumption. We explore the strategic consequences when agents possess the flexibility to strategically under-utilize their computational power.

\subsection{A New Strategic Dimension: Under-utilizing Computational Power}

We now consider a scenario where an AI agent, endowed with a certain maximum computational capacity (rationality level $\tau_i$), is not forced to use it all. Instead, it can strategically choose to play a simpler strategy, one corresponding to a period $k$ that is less than or equal to its maximum capability $\tau_i$.

\begin{definition}[Flexible strategy space]
A player $i$ with maximum rationality level $\tau_i \in \N$ has a \textbf{flexible strategy space}, denoted $A_i^{\text{flex}}$. This space consists of all eventually repeated sequences $s_i = xy^\infty$ where the length of the repeating part $y$ (the period $p_i$) satisfies $1 \le p_i \le \tau_i$.
\end{definition}

This ability to choose any level of computational expression up to one's maximum introduces strategic under-utilization. An agent might mask its true capabilities to gain a strategic advantage, making it harder for opponents to predict its actions based on its presumed (maximum) rationality level.

\subsection{The Ruins of Nash Equilibria}

This newly introduced strategic flexibility necessitates a revision of our equilibrium concepts. The potential for an agent to under-utilize its power changes the consideration of best responses.

\begin{definition}[$\epsilon$-approximate Nash equilibrium with flexible rationality]
Consider a two-player normal-form game $\G = (\{1, 2\}, \{S_1, S_2\}, \{u_1, u_2\})$. Let $\tau_1, \tau_2 \in \N$ be the maximum rationality levels of the players. Player $i$'s strategy space in the unfolding game $\G^\uf$ is $A_i^{\text{flex}}$.

A strategy profile $s = (s_1, s_2)$, where $s_i \in A_i^{\text{flex}}$, is an \textbf{$\epsilon$-approximate Nash equilibrium with flexible rationality} ($\epsilon$-NE$^{\text{flex}}$) if for each player $i \in \{1, 2\}$,
\[
\max_{s_i' \in A_i^{\text{flex}}} u_i^\uf(s_i', s_j) - u_i^\uf(s) \leq \epsilon,
\]
where $s_j$ is the strategy of the other player.
\end{definition}

Building on this, we define Nash equilibrium for almost rational agents who have this strategic flexibility.

\begin{definition}[Nash equilibria for almost rational agents with strategic flexibility]\label{def:NE-flexible}
Consider a two-player normal-form game $\G$. Let $\{\tau_1(n)\}_{n=1}^\infty$ and $\{\tau_2(n)\}_{n=1}^\infty$ be the sequences of maximum rationality levels for the two almost rational agents. For each stage $n$, player $i$'s strategy space is $A_i^{\text{flex}}(n)$.

A sequence of strategy profiles $\{s^{(n)}\}_{n=1}^\infty$, where $s^{(n)} = (s_1^{(n)}, s_2^{(n)})$ with $s_i^{(n)} \in A_i^{\text{flex}}(n)$, is a \textbf{Nash equilibrium with strategic flexibility} if for each $n$, $s^{(n)}$ is an $\epsilon_n$-NE$^{\text{flex}}$ of $\G^\uf(n)$ (the unfolding game with these flexible strategy spaces) and $\lim_{n\to\infty} \epsilon_n = 0$.
\end{definition}

With this flexibility, the conditions for reaching a no-trade outcome (a Nash equilibrium where $\epsilon_n \to 0$) dramatically change. The ability to strategically under-utilize computational power can lead to a failure to converge to a Nash equilibrium regardless of players' rationality levels. For a specific class of games, i.e., Matching Pennies games given in \eqref{eq:matching-pennies-RC}, this can be quantified.

We recall the definition of the Matching Pennies game and state the main result. A Matching Pennies game is a two-player zero-sum game where each player simultaneously chooses to show a head or a tail. If the choices match, Player $1$ wins (payoff $1$) and Player $2$ loses (payoff $0$). If they don't match, Player $2$ wins (payoff $1$) and Player $1$ loses (payoff $0$). The payoffs are binary, and for each outcome, one player receives $1$ and the other $0$.

\begin{theorem}[No reachable NEs in Matching Pennies games with flexibility]\label{thm:instability-matching-pennies}
Consider the two-player Matching Pennies game $\G$ given in \eqref{eq:matching-pennies-RC}. Let the players be almost rational, with maximum rationality levels $\{\tau_1(n)\}_{n=1}^\infty$ and $\{\tau_2(n)\}_{n=1}^\infty$, not necessarily almost identical. Then, for any sequence of strategy profiles $\{s^{(n)}\}_{n=1}^\infty$ where each $s^{(n)}$ is an $\epsilon_n$-NE$^{\text{flex}}$ in the corresponding unfolding game $\G^\uf(n)$, it holds that for every stage $n$,
\[\epsilon_n \ge \frac{1}{3}. \]
Thus, it is impossible for the almost rational agents to reach any Nash equilibrium in the sense of \Cref{def:NE-flexible} in the Matching Pennies game.
\end{theorem}

This theorem implies that for Matching Pennies games, it is impossible to reach any Nash equilibrium when the almost rational agents can strategically under-utilize their computational power. The quantity $\epsilon_n$ remains bounded away from zero. The intuition is that the possibility of strategic under-utilization introduces a profound layer of strategic uncertainty. An agent can no longer reliably infer an opponent's strategy based on their maximum computational power, as the opponent might be playing a much simpler, unpredictable (from the perspective of maximal rationality) strategy. This uncertainty hinders convergence to any NE, leading to persistent strategic maneuvering, or ``trade.''

\begin{proof}[Proof of \Cref{thm:instability-matching-pennies}]
Consider a parameter $\delta\in[0,1]$. For the $n$-th stage of the unfolding game, let player $1$ play a strategy $s_1^{(n)}$ that is a sequence of length $p_1(n)\leq \tau_1(n)$, and player $2$ play a strategy $s_2^{(n)}$ that is a sequence of length $p_2(n)\leq \tau_2(n)$. Suppose $p_1(n) \geq p_2(n)$.

Suppose player $1$'s payoff is less than $1-\delta$ in the unfolding game, meaning that the average payoff over infinite interaction is less than $1-\delta$. Then, player $1$ can choose to play a strategy $s_1'^{(n)}$ that is a sequence of length $p_1'(n) = p_2(n)$ (the same length as player $2$'s strategy). For each round $j$ in $s_1'^{(n)}$, player $1$ plays the best response against player $2$'s action, which gives $1$ payoff in that round. Since both players have the same length of strategy, this strategy is possible. By doing so, player $1$ can guarantee an average payoff of $1$ in the unfolding game, i.e., an extra $\delta$ payoff compared to the previous strategy $s_1^{(n)}$. Thus, $\epsilon_n\geq \delta$.

Suppose otherwise that player $1$'s payoff is at least $1-\delta$ in the unfolding game. Then, player $2$ obtains a payoff of at most $\delta$ in the unfolding game. Player $2$ can then ``flip'' their strategy to play the other action in each round, which guarantees a payoff of $1-\delta$ in the unfolding game, i.e., an extra $1-2\delta$ payoff compared to the previous strategy $s_2^{(n)}$. Thus, $\epsilon_n\geq 1-2\delta$.

We can choose $\delta=1/3$ to show that $\epsilon_n\geq 1/3$ for all $n$, which completes the proof.
\end{proof}
\section{Discussion}\label{sec:discussion}

Below we present several real-world implications of our results.
\begin{enumerate}
    \item \textbf{The risk of a computational ``arms race'' in AI-driven markets.} In the real world, AI agents from different companies will likely have slightly different computational power. If these agents, especially those with similar training data and thus similar prior beliefs, are made to always use their maximum computational power, our findings show a significant risk: markets could stop functioning. This ``silence of the market,'' where trades do not happen because of a constant push for more computational strength, is a serious problem for the economy and could lead to major inefficiencies.
    \item \textbf{The pursuit of perfectly homogeneous AI systems can foster dynamic market interactions.} Our main result suggests that if AI agents become perfectly identical in their computational capabilities and prior beliefs, this can prevent the market from reaching a stable Nash equilibrium in some games. This implies that striving for perfectly consistent and powerful AI systems, rather than eliminating market uncertainty, might paradoxically lead to a vibrant market characterized by continuous strategic maneuvering (i.e., ``trade''), as the absence of a stable Nash equilibrium encourages ongoing interaction.
    \item \textbf{Strategic under-utilization of computation can encourage trade.} If AI systems strategically reduce their computational effort, essentially ``pretending'' to be less powerful or conserving resources, market activity can be stimulated. This under-utilization can break computational arms race deadlocks and make agent behavior less predictable, creating trade opportunities. This suggests designing AI agents with the ability to strategically limit their computational power, which could be crucial for active and efficient markets, rather than solely pursuing maximum processing capability.
\end{enumerate}

Our work also raises several questions for future research:
\begin{enumerate}
    \item \textbf{Multi-player extensions.} Our model currently focuses only on two-player games, as is the case in the no-trade theorem. However, real-world markets involve many agents. How do our findings extend to multi-player settings? Can we still find stable no-trade outcomes when many agents with different computational powers interact?
    \item \textbf{More general result for strategic under-utilization.} We show that in the Matching Pennies game, strategic under-utilization leads to the complete breakdown of any Nash equilibrium. Can we establish a more general result regarding strategic under-utilization across different types of games? Understanding the broader implications of our findings could help in designing more robust AI systems that can adapt to various market conditions.
\end{enumerate}

\paragraph{Acknowledgments.} This work is supported by the National Science and Technology Major Project (No. 2022ZD0114904). The authors wish to express their sincere gratitude to Ariel Rubinstein, Yanjing Wang, R. Ramanujam, Changrui Mu, and Dongchen Li. Their generous engagement and many insightful discussions have significantly improved this work. The authors also thank the audience from PKU Weekly Logic Seminar 2023 and the 7th World Congress of the Game Theory Society (GAMES 2024) for their feedback.

\bibliographystyle{plain}
\bibliography{ref}

\appendix
\newpage
\section{Eventually Repeated Sequences and Computational Agents}\label{sec:eventually-repeated}

This section argues that computationally bounded agents, when generating sequences of actions, are limited to producing eventually repeated sequences. We start by introducing Turing machines as a general model of computation. We then show that Turing machines restricted to finite memory are equivalent to finite automata. Finally, we demonstrate that finite automata inherently generate eventually repeated sequences.

\subsection*{Turing Machines and the Church-Turing Thesis}

A \textbf{Turing Machine (TM)}~\cite{turingComputableNumbersApplication1937} is a foundational model of computation. It consists of an infinite tape divided into cells, a read/write head that can move along the tape, and a finite set of internal states and rules. At each step, the TM reads the symbol on the tape cell under the head, and based on this symbol and its current internal state, it writes a new symbol on the tape, changes its internal state, and moves the head one cell to the left or right. Despite its simplicity, a TM can simulate any computer algorithm. 

The \textbf{Church-Turing Thesis}~\cite{copelandChurchTuringThesis2020} posits that any computable function can be computed by a TM. This means TMs can represent any algorithmic agent, including one generating a strategy (an infinite sequence of actions) based on finite game structures.

\subsection*{Turing Machines with Finite Memory and Finite Automata}

A \textbf{Deterministic Finite Automaton (DFA)} is a simpler computational model with a finite number of states. It reads an input string and transitions between states, accepting or rejecting the string~\cite{conwayRegularAlgebraFinite2012}. While DFAs typically handle finite strings, their finite-state nature is key for understanding sequence generation with limited resources.

We are interested in agents with \emph{bounded} computational resources, particularly finite memory. While a standard TM has an infinite tape, an agent with finite memory uses only a fixed portion of this tape to decide its next action.

\begin{lemma}\label{lem:tm-finite-tape-fa}
A Turing machine using only a fixed, finite portion of its tape (say, $k$ cells) to compute an output sequence can be simulated by a Finite Automaton (FA).
\end{lemma}
\begin{proof}
The TM's operation on $k$ cells depends on its internal state, the $k$ cell contents, and its head position. Since the TM's internal states are finite, the tape alphabet is finite (so $k$ cells have finitely many content combinations), and head positions within $k$ cells are finite, the total number of distinct (TM state, $k$-cell content, head position) combinations is finite. Each such combination can be a state in an FA. The TM's transition rules and tape-writing rules together map directly to FA transition rules between these combined states. Thus, an FA can simulate the TM. Outputs can be associated with FA states or transitions.
\end{proof}

This lemma means that a computational agent with finite memory for decision-making can be modeled as a finite automaton.

\subsection*{Finite Automata and Eventually Repeated Sequences}

Consider a finite automaton that generates an infinite sequence of outputs (e.g., actions). This can be thought of as a machine (like a Moore machine~\cite{mooreGedankenexperimentsSequentialMachines1956}) that, after an initial setup (e.g., reading the game structures), produces an output at each step based on its current state and its previous output for infinitely many steps: exactly the strategies in the unfolding game.

\begin{theorem}\label{thm:fa-eventually-periodic}
Any infinite sequence of outputs generated by a deterministic finite automaton (with fixed or no ongoing input after initialization) is eventually periodic.
\end{theorem}
\begin{proof}
Imagine the finite automaton stepping through its states one by one, generating an output at each step. Let's say the automaton has a specific, finite number of states (e.g., $N$ states).
As the automaton runs, it generates a sequence of states: state at step 0, state at step 1, state at step 2, and so on.
Since there are only $N$ possible states, if the automaton runs for more than $N$ steps (say, $N+1$ steps), it must have visited at least one state more than once. Think of it like having $N$ boxes (the states) and $N+1$ items (the time steps when a state is visited); at least one box must contain two items.

So, there must be some state that the automaton was in at an earlier step (say, step $j$) and then returned to at a later step (say, step $k$).
Because the automaton is deterministic (its next state and output are uniquely determined by its current state, given that any external input is now fixed or absent), once it re-enters the state it was in at step $j$, its future behavior will be an exact repeat of its behavior after step $j$.
The sequence of states from step $k$ onwards will be the same as the sequence of states from step $j$ onwards. Consequently, the sequence of outputs from step $k$ onwards will also be the same as the sequence of outputs from step $j$ onwards.

This means the overall output sequence will look like this: an initial sequence of outputs (from step 0 up to step $j-1$), followed by a sequence of outputs (from step $j$ up to step $k-1$) that then repeats forever. This is precisely what an eventually periodic sequence is: an initial, non-repeating part, followed by a part that repeats indefinitely.
\end{proof}

In conclusion, we have the following:
\begin{enumerate}
    \item Computational agents can be modeled as Turing machines (Church-Turing thesis).
    \item TMs with finite memory for generating action sequences are equivalent to finite automata (\Cref{lem:tm-finite-tape-fa}).
    \item Finite automata generating infinite sequences produce eventually repeated sequences (\Cref{thm:fa-eventually-periodic}).
\end{enumerate}
Therefore, computationally bounded agents (with finite memory) are restricted to strategies that are eventually repeated sequences. This justifies focusing on such sequences for these agents, as discussed in \Cref{sec:unfolding-game-rationality-definition}.

\subsection*{Remark: Effect of Randomness}
For completeness, we note that the above discussion assumes deterministic automata and Turing machines. If we allow randomness in the computation, then, in principle, the agent can generate any sequence of actions, including non-eventually repeated ones. Our results would need to be adapted to account for this increased expressive power. 

However, in reality, randomness occurring in AI agents is typically pseudo-randomness (e.g., using a fixed seed), which degenerates to deterministic behavior. Since we are interested in the limit behavior of agents, the gap between deterministic and pseudo-random behavior becomes negligible in the long run. Thus, we can safely focus on deterministic agents generating eventually repeated sequences without loss of generality.
\newpage
\section{The Counterpoint Technique}\label{sec:counterpoint}

In this section, we formally develop the counterpoint technique and give detailed proofs of \Cref{thm:main-result} and \Cref{thm:correspondence}.

\subsection{Almost Coprime: Number Theory for ``Numbers at Infinity''}

To prove these two results, we need to prove a stronger result. That is, we do not assume the rationality levels of players are almost identical. The purpose of this number theory is straightforward: almost-rationality levels are ``numbers at infinity.'' Our number theory is a natural extension of classical number theory for finite numbers. For an introduction to classical number theory, see the textbook \cite{adamsIntroductionNumberTheory1976}.

We first consider the case of two numbers at infinity, i.e., the rationality levels of two players. We start with a simple observation. For two large numbers $a$ and $b$, if they are coprime, their greatest common divisor is $1$, which is very small. Such a phenomenon can be generalized to numbers that are \emph{not} coprime. For two large numbers $a$ and $b$, if their greatest common divisor is relatively small compared with $a$ and $b$ themselves, then when viewed from infinity, $a$ and $b$ are almost coprime. The precise meaning of this intuition is given in the following definition.

\begin{definition}[almost coprime]
Two sequences of natural numbers $\{a_n\}_{n=1}^\infty$ and $\{b_n\}_{n=1}^\infty$ are called \emph{almost coprime} if 
\[\lim_{n\to\infty}\frac{\gcd(a_n,b_n)}{a_n}=0\text{ and }\lim_{n\to\infty}\frac{\gcd(a_n,b_n)}{b_n}=0.\]
Or, equivalently, $\lim_{n\to\infty}\frac{\gcd(a_n,b_n)}{\min\{a_n,b_n\}}=0$. Here, $\gcd(a_n,b_n)$ is the greatest common divisor of $a_n$ and $b_n$.
\end{definition}

Under the concept of being almost identical, we can provide a more intuitive understanding of almost coprime.

\begin{proposition}\label{prop:iff-almost-coprime}
Suppose $\{a_n\}$ and $\{b_n\}$ are almost identical. Then the following two statements are equivalent:
\begin{enumerate}
    \item $\{a_n\}$ and $\{b_n\}$ are almost coprime.
    \item $a_n\neq b_n$ for all but finitely many $n$. That is, there exists $n_1\in\N$ such that for all $n\geq n_1$, $a_n\neq b_n$.
\end{enumerate}
\end{proposition}

In other words: Being almost coprime is equivalent to being almost identical but always slightly different. We will relax the constraint of being almost identical in \Cref{thm:main-result} and prove the following more general theorem:

\begin{theorem}\label{thm:main-result-almost-coprime}
Consider two players with almost-rationality levels $\{\tau_1(n)\}_{n=1}^\infty$ and $\{\tau_2(n)\}_{n=1}^\infty$. The following statements are equivalent:
\begin{enumerate}
    \item The rationality levels of both players are almost coprime.
    \item For every two-player finite normal-form game $\G$, the two players can reach a Nash equilibrium in the sense of \Cref{def:NE-for-almost-rational}.
\end{enumerate}
Moreover, \Cref{thm:correspondence} also holds for players with almost coprime rationality levels.
\end{theorem}

Importantly, for different directions of the proof, our goal is different.
\begin{itemize}
    \item Concerning the sufficiency proof of \Cref{thm:main-result-almost-coprime}, for \emph{each NE in each finite normal-form game}, we want to \emph{construct} a series of strategy profiles and show that they are $\epsilon_n$-NEs with $\epsilon_n\to 0$ as $n\to\infty$. Thus, we only need to have the ability to calculate the approximation $\epsilon_n$ of our constructed strategy profiles.
    \item Concerning the necessity proof of \Cref{thm:main-result-almost-coprime}, we consider the contrapositive statement. That is, we want to show that if players are not almost coprime, then there \emph{exists a NE of a finite normal-form game} such that the approximation of \emph{any} strategy profile cannot be an $\epsilon_n$-NE with $\epsilon_n\to 0$ as $n\to\infty$. Thus, we need to show that the approximation of any strategy profile is large enough.
\end{itemize}

Quite intuitively, the necessity proof is much harder than the sufficiency proof since we must be able to calculate an appropriate lower bound for the approximation of \emph{any} strategy profile. Nevertheless, even the sufficiency proof is not easy, since we need at least to know how to \emph{compare} the payoff of the constructed strategy profiles with that of \emph{any} other strategy profile. Moreover, it seems less intuitive how we can come up with the condition of being almost coprime. Below, we gradually build up our intuition and develop the counterpoint technique.

\subsection{Preparations}

Below we use $\I=\{1,2\}$ to represent the player set. We use the notation $f_\G(s)$ to denote the maximum deviation of a strategy profile $s$ in the game $\G$, where $\G$ could be either a finite normal-form game or an unfolding game. We use $\fld(s)$ to denote the folding mixed strategy of $s$ in the original game $\G$.

A very important fact is that it suffices to consider a finite fragment of strategies (viewed as infinite sequences).

\begin{lemma}\label{lemma:only-need-repeated}
Let $\G=(\I,\{S_i\}_{i\in \I},\{u_i\}_{i\in \I})$ be a finite normal-form game and $\G^\uf$ be its unfolding. Then there exists $\tau=\lcm(\tau_i)_{i\in \I}$ such that for any sufficiently large $k$ and any strategy profile $s=(s_i)_{i\in \I}$ of $\G^\uf$,
\[u_i^\uf(s)=\frac{1}{\tau}\sum_{j=k+1}^{k+\tau} u_i(s(j))\quad \text{for each }i\in \I.\]
Moreover, suppose $\mu=\fld(s)$. Then for each $i\in \I$ and $a\in S_i$,
\[\mu_i(a)=\frac{1}{\tau}|\{k+1\leq j\leq k+\tau:s_i(j)=a\}|.\]
\end{lemma}

\begin{proof}
    The proof mainly relies on the repeated structure of the strategy profiles. Combining the definition of $u^\uf$ and $\mu$ with the basic property of limits, we know that for every fixed $k\geq 1$ and strategy profile $s$, we can calculate $\mu$ and $u^\uf$ by
    \begin{equation}
        u_i^\uf(s)=\lim_{n\to\infty}\frac{1}{n-k+1}\sum_{j=k}^n u_i(s(j))\quad\text{and}\quad \mu_i(a) = \lim_{n \to \infty} \frac{1}{n-k+1}\sum_{j=k}^n\chi_{\{s_i(j)=a\}}(j).\label{eq:alter-cal}
    \end{equation}

    Suppose player $i$'s rationality level is $\tau_i$. We claim that $\tau=\lcm(\tau_i)_{i\in \I}$ has the desired property. Note that every strategy is eventually periodic. Then there exists a $k_0$ such that for all $i\in \I$, $(s_i(j))_{j\geq k_0}$ is a periodic infinite sequence. We take any $k\geq k_0$. Clearly, after the $k$-th round, every player's strategy has the form $(x_1\dots x_{\tau_i})^\infty$ with each $x_j$ an action in $S_i$. By repeating $x_1\dots x_{\tau_i}$ $\tau/\tau_i$ times, the strategy can also be written as 
    \[(\underbrace{x_1\dots x_{\tau_i}x_1\dots x_{\tau_i}\dots x_1\dots x_{\tau_i}}_{\tau/\tau_i\text{ times}})^\infty=(y_1\dots y_\tau)^\infty.\] 
    Now, the strategy profile after the $k$-th round has the form:
    \[\begin{array}{cccccc}
        s_1:&s_1(k+1)&s_1(k+2)&\cdots&s_1(k+\tau)&\cdots\\
        s_2:&s_2(k+1)&s_2(k+2)&\cdots&s_2(k+\tau)&\cdots\\
        \vdots&\multicolumn{5}{c}{\cdots}\\
        s_{|\I|}:&s_{|\I|}(k+1)&s_{|\I|}(k+2)&\cdots&s_{|\I|}(k+\tau)&\cdots\\
    \end{array}\]
    Thus, viewing $s$ as an infinite sequence, we know that $s$ itself is eventually periodic with period $\tau$ after the $k$-th round. Then, as a function of $j$, $u_i(s(j))$ and $\chi_{\{s_i(j)=a\}}(j)$ are also eventually periodic. Then, by the basic property of limits, the lemma follows from \eqref{eq:alter-cal}.
\end{proof}

Then we define the basic concepts used in the counterpoint technique. Let us use music to make a metaphorical description. Strategies are voices sung by different players. The key feature of strategies is that they are repetitions of a basic melody. Different melodies have different characteristics. Thus, the goal of a composer is to write melodies so that different voices progress in harmony. A typical technique used for such music is called \emph{counterpoint}. Thus, we also call our technique counterpoint. Formally, we have the following notions.

\begin{definition}[notes, melodies, chords, and pieces]
Let $\G=(\I,\{S_i\}_{i\in \I},\{u_i\}_{i\in \I})$ be a finite normal-form game. Let $\tau_i$ be the rationality level of player $i$. Consider the unfolding $\G^\uf$ with rationality level $(\tau_i)_{i\in \I}$. Consider a strategy profile $s$ of $\G^\uf$. We define the following concepts:
\begin{itemize}
    \item The \emph{melody} $m_i$ of $s_i$ ($i\in \I$) is the sequence $s_i(1)s_i(2)\dots s_i(\tau_i)$.
    \item The $j$-th \emph{note} of the melody $m_i$ is $m_i(j)$, i.e., the $j$-th element of $m_i$. We view notes occurring in different positions in a melody as different notes. That is, even if $m_i(j_1)=m_i(j_2)$, we \emph{do not} regard them as the same note unless $j_1=j_2$.
    \item A \emph{chord} is a pure strategy profile of $\G$ given by every player's note in the melody. We also consider the positions of notes when comparing two chords.
    \item The \emph{piece} $p$ of $s$ is the sequence $s(1)s(2)\dots s(\tau)$, where $\tau$ is defined as in \Cref{lemma:only-need-repeated}, i.e., $\tau=\lcm(\tau_i)_{i\in \I}$.
\end{itemize}
\end{definition}

When a melody is chosen, the corresponding repeated strategy can be constructed in an obvious way. Thus, when all players have chosen their melodies, chords at each round and the piece are uniquely determined. In view of \Cref{lemma:only-need-repeated}, we only need to consider all pieces of length $\tau$ such that every player $i$ produces a melody of length $\tau_i$.

To simplify the following discussion, suppose without loss of generality that $u_i\geq 0$ for all $i\in \I$.

Now fix a finite normal-form game $\G=(\I,\{S_i\}_{i\in \I},\{u_i\}_{i\in \I})$ and its Nash equilibrium $\sigma_*$. We discuss what kinds of melodies can converge to $\sigma_*$ in $\G$. This will establish the sufficiency of \Cref{thm:main-result-almost-coprime}. The necessity will be proved in \Cref{subsec:modified-matching-pennies}.

We can understand harmony in games by its reverse. We view the level of dissonance as the maximum deviation of a certain part. When two melodies start at different moments, different notes will occur at the same moment. In an infinite musical piece, we cannot expect permanent harmony in its progress. We have to choose melodies so that they are always not too dissonant.

Intuitively, \emph{the simpler the melodies are, the more likely harmony can be reached}. The simplest melody is one where every note appears only once but with different durations. For example, suppose player 1 is writing a basic melody of length $6$; then the simplest one is $HHHTTT$. Writing $HHHTTT$ gives more chances to reach harmony than writing $HTTHTH$. More generally, to converge to an NE, players should play strategies in a similar manner: First, play action $a_1$ for a $p_1$ proportion of time, then play action $a_2$ for a $p_2$ proportion of time, and so on. We call it the \emph{simple melody}. Formally, we have the following definition.

\begin{definition}[simple melodies]
A melody of player $i$ is called \emph{simple} if it has the form 
\[a_1^{n_1}a_2^{n_2}\dots a_r^{n_r}\]
for distinct actions $a_k$.
\end{definition}

Now we formally construct the simple melodies used for converging to the NE $\sigma_*$. Consider player $i$ with a rationality level of $\{\tau_i(n)\}_{n\geq 1}$. Consider actions in $\supp(\sigma_{*i})$, enumerated by $a_{i1},\dots,a_{ir_i}$. By rational approximation, for each $a_{ik}$ ($1\leq k\leq r_i-1$), there exists a natural number sequence $\{\nu_{i,k}(n)\}_{n\geq 1}$ such that

\begin{equation}
    \lim_{n\to\infty}\frac{\nu_{i,k}(n)}{\tau_i(n)}=\sigma_{*i}(a_{ik}).\label{eq:simple-melody-freq}
\end{equation}

Let $\nu_{i,{r_i}}(n)=\tau_i(n)-\sum_{k=1}^{r_i-1}\nu_{i,k}(n)$. Then it is not hard to show that \eqref{eq:simple-melody-freq} also holds for $k=r_i$. Now, for the $n$-th rationality level, player $i$ can choose the following simple melody:

\begin{equation}
    m_i^\s{n}=a_{i1}^{\nu_{i,1}(n)}a_{i2}^{\nu_{i,2}(n)}\dots a_{ir_i}^{\nu_{i,r_i}(n)}.\label{eq:simple-melody}
\end{equation}

We directly have the following property.

\begin{lemma}\label{lemma:simple-approach}
Suppose each player $i$ chooses a simple melody $m_i$ as in \eqref{eq:simple-melody}. Construct a strategy profile $s^n$ using these melodies. Then, $\lim_{n\to\infty}\fld\left(s^n\right)=\sigma_*$.
\end{lemma}

By this lemma, the main crux of the proof for the main result is to establish the NE convergence in the unfolding, that is, $\lim_{n\to\infty}f_{\G^\uf(n)}\left(s^\s{n}\right)=0$.

\subsection{A Reduction to Coprime Cases}\label{subsec:reduce-coprime}
In this part, we prove the sufficiency of \Cref{thm:main-result-almost-coprime} and \Cref{thm:correspondence}. Our method here is actually to reduce the general cases to the coprime cases. 

\subsubsection{Coprime Cases}
We first show that any NE in any finite normal-form game can be realized by the folding strategy when two players have coprime rationality levels at each time. Formally, we prove the following proposition.

\begin{proposition}\label{prop:coprime-approaching}
Suppose that $\{\tau_i(n)\}_{n\geq 1}$ is the rationality level of player $i$ ($i=1,2$) so that $\tau_1(n),\tau_2(n)$ are coprime for each $n$. Then strategy profiles $s^\s{n}$ constructed by the simple melody $m_i^\s{n}$ in \eqref{eq:simple-melody} converge to a Nash equilibrium in the unfolding game, and the equilibrium payoff is the same as the payoff from $\sigma^*$ in the original game.
\end{proposition}

The key property of coprime cases, as demonstrated in \Cref{sec:proof-idea}, is that the payoff and the approximation of a strategy profile in the unfolding are identical to those of its folding in the original game.

\begin{lemma}\label{lemma:coprime-cal-f_val}
Suppose that $\tau_1$ and $\tau_2$ are the coprime rationality levels of players $1$ and $2$, respectively. Consider the unfolding $\G^\uf$ of $\G$. Then, for any strategy profile $s$ of $\G^\uf$ constructed by melodies $m_1$ and $m_2$, we have
\[u^\uf(s)=u(\fld(s))\quad\text{and}\quad f_{\G^\uf}(s)=f_{\G}(\fld(s)).\]
\end{lemma}

\begin{proof}
    Since $\tau_1, \tau_2$ are coprime, by the property of greatest common divisor, the length $\tau$ of piece $p$ is exactly $\tau_1\tau_2$. The key observation is that every possible chord occurs exactly once in a piece. More precisely, consider the chord $c_{j,k}$ formed by the $j$-th note $m_1(j)$ of player $1$ and the $k$-th note $m_2(k)$ of player $2$. Then we claim that $c_{j,k}$ occurs exactly once in $p$. This is, in fact, a basic group-theoretic property of the multiplicative group of residue classes. Below, we present an elementary proof for completeness.

    First, $c_{j,k}$ must occur at most once. Suppose this is not the case; then there are two time steps $t_1<t_2$ such that $1\leq t_i\leq\tau$, and $p(t_1)$ and $p(t_2)$ are the same chord. Then we must have $\tau_i\mid t_2-t_1$ for $i=1,2$. Then, by the definition of least common multiple, $\tau\mid t_2-t_1$. However, $1\leq t_2-t_1\leq \tau-1$, which is impossible.
    
    Then we show that $c_{j,k}$ must occur at least once. This follows by basic counting: There are $\tau$ chords in $p$, and there are at most $\tau_1\tau_2=\tau$ possible different chords in $p$, each occurring at most once. Thus, every possible chord must occur once.
    
    Now we know the chord structure in the piece $p$. Then, by \Cref{lemma:only-need-repeated}, we only need to calculate $u_i^\uf\left(s\right)$ as the average payoff in the piece $p$. That is, the average payoff of all chords:
    \[
    u_i^\uf\left(s\right)=\frac{1}{\tau}\sum_{j=1}^{\tau_1}\sum_{k=1}^{\tau_2}u_i(m_1(j),m_2(k)).
    \]
    Enumerate the actions of player $i$ as $a_{i1},a_{i2},\dots,a_{ij_i}$. Now we re-index the sum by actions, and we have:
    \[
    u_i^\uf\left(s\right)=\sum_{j=1}^{j_1}\sum_{k=1}^{j_2}\frac{|t\in[1,\tau_1]:m_1(t)=a_{1j}|}{\tau_1}\cdot\frac{|t\in[1,\tau_2]:m_2(t)=a_{2k}|}{\tau_2}\cdot u_i(a_{1j},a_{2k}).
    \]
    Then, by \Cref{lemma:only-need-repeated}, the above sum is exactly $u_i(\fld(s))$.
    
    Now we consider $f_\G$ and $f_{\G^\uf}$. Let $\mu=\fld(s)$. By a more careful inspection, when we calculate $u_1^\uf(s)$, every note $m_1(j)$ is \emph{as if} facing a mixed strategy $\mu_2$. That is, when we choose $m_1(j)=a_1'$, player $1$'s payoff in the unfolding contributed by $m_1(j)$ is:
    \[
    \frac{1}{\tau}\sum_{k=1}^{\tau_2}u_i(a_1',m_2(k))=\frac{1}{\tau_1}u_1(a_1',\mu_2).
    \]
    Alternatively, when we change $m_1(j)$ to $a_1'$, the payoff change in the unfolding is:
    \[
    \frac{1}{\tau}\sum_{k=1}^{\tau_2}(u_i(a_1',m_2(k))-u_i(m_1(j),m_2(k)))=\frac{1}{\tau_1}(u_1(a_1',\mu_2)-u_1(m_1(j),\mu_2)).
    \]

    The payoff change caused by the deviation of each note is independent. Thus to calculate the maximum extra gain of player $1$ from deviations, it suffices to choose every note as the best response $a^*_1$ against $\mu_2$. Similarly, we can choose the best response $a_2^*$ of player $2$ against $\mu_1$. These lead to strategies in unfolding: $s_i^*=(a_i^*)^\infty$. Then we have
    \begin{align*}
        f_{\G^\uf}(s)&=\max_{i=1,2}\{u_i^\uf(s^*_i,s_{-i})-u_i^\uf(s_i,s_{-i})\}\\
                    &=\max_{i=1,2}\{u_i(a^*_i,\mu_{-i})-u_i(\mu_i,\mu_{-i})\}\\
                    &=f_\G(\mu).\qedhere
    \end{align*}
\end{proof}

Now we are ready to prove \Cref{prop:coprime-approaching}.
\begin{proof}[Proof of \Cref{prop:coprime-approaching}]
    Let $\mu^\s{n}=\fld\left(s^\s{n}\right)$. By \Cref{lemma:only-need-repeated}, 
    \[\mu^\s{n}_i(a_{ik})=\frac{\nu_{i,k}(n)}{\tau_i(n)}.\]
    By \eqref{eq:simple-melody-freq},
    \[\lim_{n\to\infty}\norm{\mu^\s{n}-\sigma_*}=0.\]

    Note that the functions $u$ are continuous. Thus, $\lim_{n\to\infty} u(\mu^\s{n}) = u(\sigma_*)$. The result then directly follows from \Cref{lemma:coprime-cal-f_val}.
\end{proof}

\subsubsection{General Cases}
Now we turn to the general cases. To reduce the general cases to the coprime cases, we need to ``bundle up'' notes. Let us again consider a simple scenario of the Matching Pennies game. Suppose player 1's level of almost-rationality is $\{2k\}_{k\geq 1}$ and player 2's is $\{2k+2\}_{k\geq 1}$. Then if we view every two notes as a bundle, forcing player 1 and player 2 to choose the same action in a bundle, then they are acting \emph{as if} with levels of almost-rationality of $\{k\}_{k\geq 1}$ and $\{k+1\}_{k\geq 1}$. Now the situation becomes exactly the coprime one.

Formally, we introduce the concept of \emph{bundles}.
\begin{definition}[bundles]
Suppose two players have rationality levels $\tau_1$ and $\tau_2$. They choose melodies $m_1$ and $m_2$, respectively. Let $\rho=\gcd(\tau_1,\tau_2)$. A \emph{bundle} of player $i$ is a subsequence of $m_i$ in the form $m_i(1+k\cdot \rho)m_i(2+k\cdot \rho)\dots m_i(\rho+k\cdot \rho)$ for some $k\geq 0$.
\end{definition}

Now, we need to establish a result similar to \Cref{lemma:coprime-cal-f_val}. To do so, we need a concept resembling folding as follows.

\begin{definition}[bundle folding]
    Suppose two players have rationality levels $\tau_1$ and $\tau_2$. They choose melodies $m_1$ and $m_2$, respectively. Let $\rho=\gcd(\tau_1,\tau_2)$. The \emph{bundle folding} of $m_i$ (or a strategy constructed by $m_i$) is defined to be a series of $\rho$ mixed strategies $\{\mu_{j,i}\}_{j=1}^\rho$ in $\G$ such that
    \[\mu_{j,i}(a)=\frac{1}{\tau_i/\rho}|\{t\in[0,\tau_i/\rho-1]:m_i(j+t\cdot \rho)=a\}|.\]
    Moreover, we call $\mu_j=(\mu_{j,1},\mu_{j,2})$ \emph{the $j$-th bundle profile}.
\end{definition}

Intuitively, bundle folding is the result of folding strategies bundle by bundle. 

Let $f^i_\G$ is the maximum extra payoff of player $i$ in $\G$ by deviating from the current strategy. Then we have the following lemma analogous to \Cref{lemma:coprime-cal-f_val}.
\begin{lemma}\label{lemma:no-coprime-cal-f-val}
Suppose that $\tau_1, \tau_2$ are the rationality levels of players $1$ and $2$, respectively. Consider the unfolding $\G^\uf$ of $\G$. Let $\rho=\gcd(\tau_1,\tau_2)$. Take any strategy profile $s$ of $\G^\uf$ constructed by melodies $m_1$ and $m_2$. Suppose the bundle profiles of the two players are $\mu_1,\mu_2,\dots,\mu_\rho$, respectively. Then we have
\[u_i^\uf(s)=\frac{1}{\rho}\sum_{j=1}^\rho u_i(\mu_j).\]
Moreover,
\[f_{\G^\uf}(s)=\frac{1}{\rho}\max_{i=1,2}\left\{\sum_{j=1}^\rho f^i_\G(\mu_j)\right\}.\]
\end{lemma}

\begin{proof}
    We calculate the average payoff note by note in the bundle. Consider player $1$ and her first bundle $b_{11}=m_1(1)\dots m_1(\rho)$.\footnote{We make this assumption to simplify the notation. The discussion below can be applied to any bundle in an obvious way.} We calculate payoffs contributed by $m_1(j)$. Let $\rho(n)=\gcd(\tau_1(n),\tau_2(n))$. By the same argument in the proof of \Cref{lemma:coprime-cal-f_val}, every bundle of player $1$ encounters every bundle of player $2$ exactly once in the piece $p$. Then the notes of player $2$ that $m_1(1)$ encounters are
    \begin{equation}
        m_2(j),m_2(j+\rho),\dots,m_2\left(j+\left(\tau_2/\rho-1\right)\cdot \rho\right),  \label{eq:encounter-seq}
    \end{equation}
    each encountered once.
    
    Thus, the payoff contributed by $m_1(1)$ to $u_1^\uf(s)$ is
    \[\frac{1}{\tau_1}\sum_{k=0}^{\tau_2/\rho-1}\frac{1}{\tau_2/\rho}\cdot u_1\left(m_1(j),m_2(j+k\cdot\rho)\right).\]
    Similar to the proof of \Cref{lemma:coprime-cal-f_val}, we can re-index the sum by actions. By the definition of bundle folding, we have 
    \[\frac{1}{\tau_1}\sum_{k=0}^{\tau_2/\rho-1}\frac{1}{\tau_2/\rho}\cdot u_1\left(m_1(j),m_2(j+k\cdot\rho)\right)=\frac{1}{\rho}\cdot\frac{1}{\tau_1/\rho}u_1(m_1(j),\mu_{j,2}).\]
    Summing over all notes in a bundle for player $1$. The left-hand side is the payoff contributed by the $j$-th position of all bundles. Omitting the first term $1/\rho$, the right-hand side is exactly $u_1(\mu_{j,1},\mu_{j,2})$. Thus, the actual right-hand side is exactly $(1/\rho)\cdot u_1(\mu_j)$.

    Now, sum up with $j$ from $1$ to $\rho$. Then the equation for the payoff becomes
\[u_1^\uf(s)=\frac{1}{\rho}\sum_{j=1}^\rho u_1(\mu_j).\]
That is the desired property in the lemma.

    Now we turn to $f_\G$. The proof is almost identical to that in \Cref{lemma:coprime-cal-f_val}. We only sketch the proof here. The main difference is that now players are as if playing game $\G$ for $\rho$ times, with each bundle profile occurring once. Thus, when we consider their best responses, each position of the bundle should be considered separately. For the bundle profile $j$, the maximum extra payoff of player $i$ in $\G$ it contributes is exactly $f^i_\G(\mu_j)$. Then, to maximize her payoff in the unfolding $\G^\uf$, player $i$ should choose the best response $a_{i,j}^*$ at the $j$-th position of all bundles. Her best response against $s_{-i}$ in $\G^\uf$ is then $(a_{i1}a_{i2}\dots a_{i\rho})^\infty$.
\end{proof}

Crucially, we need to control the \emph{aggregate} $f^i$ in the bundle for each player $i$. To compute $\mu_j$, we calculate the frequency vector of each player in the bundle. For player $i$, we calculate the frequency vector using $\tau_i/\rho$ notes, which determines the minimum precision $\mu_j$ can achieve. Here is where the almost-coprime condition comes into play. To make $\mu_j$ arbitrarily close to an NE, we need to make $\rho/\tau_i$ arbitrarily close to $0$ for each player $i$. This is precisely the almost-coprime condition!

Note that the calculation of bundle folding $\mu_j$ does not involve all notes, as the folding in \Cref{lemma:coprime-cal-f_val} does. Instead, it ``samples'' notes and calculates their frequency vector. Thus, we need to show that the simple melody construction can still work even if we only consider these sampled notes.

To see how this can be done, we take an NE $\sigma^*$ in the general form. Then, when the rationality levels of both players are almost coprime, we have $\rho/\tau_i\to 0$. However, the frequency of each action $a_i$ of player $i$ tends to $\sigma^*_i(a_i)$, which is a positive constant. Thus, the sampled notes are quite uniformly distributed. This is the key observation that makes the simple melody construction work. A quite intuitive illustration is given in \Cref{fig:almost-coprime-uniform-sample}.

\begin{figure}[ht]
    \centering
    \includegraphics[width=\linewidth]{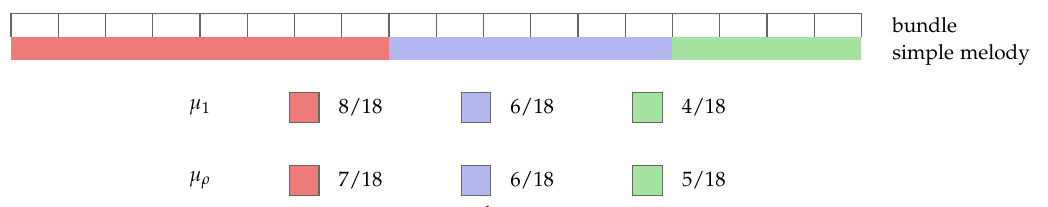}
    \caption{Illustration of the simple melody construction in almost-coprime cases. We show the frequency vector of a player in $\mu_1$ and $\mu_\rho$. Note that they are only slightly different.}
    \label{fig:almost-coprime-uniform-sample}
\end{figure}

Such a ``uniform sampling'' property can be formalized as follows.

\begin{lemma}\label{lemma:coprime-num-hetero-bundle}
Suppose two players have almost-coprime almost-rationality levels $\{\tau_i(n)\}_{n\geq 1}$, $i=1,2$. A strategy profile $s^\s{n}$ is constructed by the simple melody $m_i^\s{n}$ in \eqref{eq:simple-melody}. Suppose there are $d_i(n)$ bundles in $m_i^\s{n}$. For every $a_{ij}$ in $\supp(\sigma_{*i})$, let $d_{ij}(n)$ denote the number of bundles composed solely of $a_{ij}$. Then, for sufficiently large $n$,
\[\tau_i(n)-\sum_{a_{ij}\in\supp(\sigma_{*i})}d_{ij}(n)\leq r_i-1.\]
Here, $r_i$ is the number of actions in $\supp(\sigma_{*i})$. Equivalently, there are at most $r_i-1$ heterogeneous bundles in $m_i^\s{n}$. Moreover, 
\[\lim_{n\to\infty}\frac{d_{ij}(n)}{d_i(n)}=\sigma_{*i}(a_{ij})>0.\]
\end{lemma}
    
\begin{proof}
    By symmetry, we only prove the lemma for player $1$. By the definition of almost coprime, $d_1(n)=\tau_1(n)/\gcd(\tau_1(n),\tau_2(n))\to\infty$ as $n\to\infty$. We can suppose $n$ is large so that every action $a_{1j}$ supported in $\sigma_{1*}$ occupies at least three bundles of notes. That is, $\nu_{1,j}(n)/\gcd(\tau_1(n),\tau_2(n))\geq 3$. This is possible as we argue below. 

    Note that $\nu_{1,j}(n)/\tau_1(n)\to\sigma_{*1}(a_{1j})(>0)$ as $n\to\infty$. Thus, 
    \[\frac{\nu_{1,j}(n)}{\gcd(\tau_1(n),\tau_2(n))}=\underbrace{\frac{\nu_{1,j}(n)}{\tau_1(n)}}_{\to\sigma_{*1}(a_{1j})>0}\cdot\underbrace{\frac{\tau_1(n)}{\gcd(\tau_1(n),\tau_2(n))}}_{\to\infty}\to\infty.\]
    Therefore, for large enough $n$, $\nu_{1,j}(n)/\gcd(\tau_1(n),\tau_2(n))\geq 3$ for every $j$.
    
    By the division algorithm, action $a_{11}$ occupies $\lfloor\nu_{1,j}(n)/\gcd(\tau_1(n),\tau_2(n))\rfloor\geq 3$ bundles. The remaining unbundled notes of action $a_{11}$ are then bundled with action $a_{12}$. By our choice of $n$, the leftover $a_{12}$ notes can still form at least one bundle. We bundle them into as many bundles as possible. Then the remaining unbundled notes of $a_{12}$ are bundled with action $a_{13}$. This procedure continues until all notes are in a bundle.
    
    Consider the subsequence $a_{1j}^{\nu_{1,j}(n)}$ of $m^\s{n}$. Note that by our construction, the only heterogeneous bundles must be at the front or the tail of $a_{1,j}^{\nu_{1,j}(n)}$. More precisely, the heterogeneous bundles have the form $a_{1j}^{n_j}a_{1,j+1}^{n_{j+1}}$, with each $j$ forming one. Thus, there are at most $r_i-1$ heterogeneous bundles.

    Now we turn to the second part of the lemma. The only heterogeneous bundles are at the front or the tail of $a_{1,j}^{\nu_{1,j}(n)}$, which at most consume $2\gcd(\tau_1(n),\tau_2(n))$ notes. Thus we have that
    \[\frac{\nu_{1,j}(n)}{\gcd(\tau_1(n),\tau_2(n))}-2\leq d_{1,j}(n)\leq \frac{\nu_{1,j}(n)}{\gcd(\tau_1(n),\tau_2(n))}.\]
    Consequently, by the definition of $d_1(n)$, we have
    \begin{gather*}
        \frac{1}{d_1(n)}\cdot\left(\frac{\nu_{1,j}(n)}{\gcd(\tau_1(n),\tau_2(n))}-2\right)\leq \frac{d_{1,j}(n)}{d_1(n)}\leq \frac{1}{d_1(n)}\cdot\left(\frac{\nu_{1,j}(n)}{\gcd(\tau_1(n),\tau_2(n))}\right)\\
        \iff
        \frac{\nu_{1,j}(n)}{\tau_1(n)}-\frac{2}{d_1(n)}\leq \frac{d_{1,j}(n)}{d_1(n)}\leq \frac{\nu_{1,j}(n)}{\tau_1(n)}.
    \end{gather*}
    Letting $n\to\infty$, the lemma follows by \eqref{eq:simple-melody-freq} and the definition of almost coprime.
\end{proof}

Using this lemma, we can show that the bundle profiles can uniformly converge to an NE.

\begin{corollary}\label{cor:coprime-bundle-NE-approaching}
Suppose two players have almost-coprime levels of rationality $\{\tau_i(n)\}_{n\geq 1}$, $i=1,2$. Let $\rho(n)=\gcd(\tau_1(n),\tau_2(n))$. Bundle strategy profiles $\mu_j^\s{n}$ ($j=1,\dots,\rho(n)$) are constructed by the simple melody $m_i^\s{n}$ in \eqref{eq:simple-melody}. Then we have
\[\lim_{n\to\infty}\max_{1\leq j\leq\rho(n)}\norm{\mu_j^\s{n}-\sigma_*}=0.\]
\end{corollary}

\begin{proof}
    Consider any $\epsilon>0$. By \Cref{lemma:coprime-num-hetero-bundle}, there exists $N\in\N$ such that for all $a_{ik}$ in $\supp(\sigma_{*i})$ and $n\geq N$,
    \begin{equation}
        \left|\frac{d_{ik}(n)}{d_i(n)}-\sigma_{*i}(a_{ik})\right|<\epsilon,\label{eq:bundle-folding-approx}
    \end{equation}
    and $\tau_i(n)-\sum_{a_{ij}\in\supp(\sigma_{*i})}d_{ij}(n)\leq r_i-1$. Here we use $d_i$ and $d_{ik}$ as stated in \Cref{lemma:coprime-num-hetero-bundle}.
    
    By the definition of almost coprime, $d_i(n)\to\infty$ as $n\to\infty$. We pick $N'\geq N$ such that for all $n\geq N'$, $d_i(n)>(r_i-1)/\epsilon$.
    
    Then by \Cref{lemma:coprime-num-hetero-bundle} and the definition of bundle folding, for all $j\in[1,\rho(n)]$,
    \begin{align}
        &d_{ik}(n)\leq |\{t\in[0,d_i(n)-1]:m_i(j+t\cdot\rho(n))=a_{ik}\}|\leq d_{ik}(n)+r_i-1\notag\\
        \iff&\frac{d_{ik}(n)}{d_i(n)}\leq\mu_{j,i}^\s{n}(a_{ik})\leq\frac{d_{ik}(n)}{d_i(n)}+\frac{r_i-1}{d_i(n)}.\label{eq:est-bundle-folding}
    \end{align}
    Combining \eqref{eq:bundle-folding-approx} and \eqref{eq:est-bundle-folding}, we have
    \begin{align*}
        -\epsilon<\mu_{j,i}^\s{n}(a_{ik})-\sigma_{*i}(a_{ik})<\epsilon+\frac{r_i-1}{d_i(n)}<2\epsilon.
    \end{align*}
    Since this holds for all $i$, all $a_{ik}$, and all $j$, we have
    \[\max_{1\leq j\leq\rho(n)}\norm{\mu_j^\s{n}-\sigma_*}<2\epsilon.\]
    Since $\epsilon$ is arbitrary, the lemma follows by the definition of a limit.
\end{proof}

Now we are ready to present the following proposition, parallel to \Cref{prop:coprime-approaching}, which immediately proves the sufficiency of \Cref{thm:main-result-almost-coprime}.

\begin{proposition}\label{prop:2player-approaching}
Suppose that $\{\tau_i(n)\}_{n\geq 1}$ are the almost-coprime rationality levels of player $i$ ($i=1,2$). Then the strategy profiles $s^\s{n}$ constructed by the simple melody $m_i^\s{n}$ in \eqref{eq:simple-melody} converge to a NE in the unfolding game with the same payoff as $\sigma^*$ in the original game.
\end{proposition}
\begin{proof}
Let $\rho(n)=\gcd(\tau_1(n),\tau_2(n))$. Let $\mu_j^\s{n}$ ($j=1,\dots,\rho(n)$) be the bundle profiles of $s^\s{n}$. By \Cref{cor:coprime-bundle-NE-approaching},
\[\lim_{n\to\infty}\max_{1\leq j\leq\rho(n)}\norm{\mu_j^\s{n}-\sigma_*}=0.\]
Note that the functions $u$ and $f_\G^i$ are continuous. Thus, for player $i\in\{1,2\}$, we have
\[\lim_{n\to\infty} \max_{1\leq j\leq\rho(n)}\left|u_i\left(\mu_j^\s{n}\right)-u_i(\sigma_*)\right|=0\quad\text{and}\quad\lim_{n\to\infty} \max_{1\leq j\leq\rho(n)}f_\G^i\left(\mu_j^\s{n}\right)=0.\]
The result follows directly from \Cref{lemma:no-coprime-cal-f-val}.\footnote{The convergence results must be uniform in $j$ (i.e., the notation $\max_{1\leq j\leq\rho(n)}$) to make a uniform estimation of the limit. We here omit the process of estimation. The reader may refer to the proof of \Cref{cor:coprime-bundle-NE-approaching} for a similar argument.}
\end{proof}

\subsection{The Modified Matching Pennies Games}\label{subsec:modified-matching-pennies}
In this part, we prove the necessity of \Cref{thm:main-result-almost-coprime}. To prove these results, we need to construct a game with a non-approachable NE. Surprisingly, the constructions are all very simple modifications of the Matching Pennies game \eqref{eq:matching-pennies-RC}.

Now we present the modified Matching Pennies game (after proper scaling and shifting). Given any $\delta\in[0,1]$, the desired game $\G_\delta$ has payoff matrices
\begin{equation}\label{eq:modified-matching-pennies-RC}
    R:\begin{array}{c|cc}
            &H&T  \\\hline
         H& \delta+1&1\\
         T&2\delta&\delta+1\\
    \end{array}\qquad
    C:\begin{array}{c|cc}
            &H&T  \\\hline
         H& 2\delta&\delta+1\\
         T&\delta+1&1\\
    \end{array}
\end{equation}

Now we prove a continuity property of the function $f_\G(r,c)$, which is the maximum extra payoff of player $1$ in game $\G_\delta$ by deviating from strategy profile $(r,c)$.
\begin{lemma}\label{lemma:reverse-f_G-Lipschitz}
Game \eqref{eq:modified-matching-pennies-RC} has a unique NE $r^*=c^*=(\delta,1-\delta)^\T$. Moreover, for any $\delta\in[0,1]$ and $\epsilon>0$, there exists some $\epsilon_0>0$ such that $f_\G(r,c)\geq \epsilon_0$ whenever $\norm{(r,c)-(r^*,c^*)}\geq \epsilon$.
\end{lemma}

\begin{proof}
    It is not hard to check that no pure strategy can appear in any NE. Thus, both players must use non-pure strategies. When player 2 chooses $(\delta,1-\delta)$, player 1's expected payoffs are 
    \[\begin{array}{c|c}
        &\delta H+ (1-\delta) T\\\hline
        H&\delta(\delta+1)+(1-\delta)\cdot 1=\delta^2+1\\
        T&\delta\cdot 2\delta+(1-\delta)\cdot(\delta+1)=\delta^2+1\\
    \end{array}\]
    Thus, both $H$ and $T$ are best responses for player 1. Therefore, $r^*$ is the best response against $c^*$. Similarly, $c^*$ is the best response against $r^*$. Thus, $(r^*,c^*)$ is an NE. When player 2 deviates from $c^*$, player 1's best response is solely $H$ or solely $T$. Only when $c=c^*$ can player 1 have a non-pure best response. A similar argument holds for the other case. Thus, $(r^*,c^*)$ is the unique NE.

    Now, suppose $\norm{(r,c)-(r^*,c^*)}\geq \epsilon$. Consider the compact set $K=\{(r,c)\in\Delta_2\times\Delta_2:\norm{(r,c)-(r^*,c^*)}\geq \epsilon\}$. Since $f_\G(r,c)>0$ when $\norm{(r,c)-(r^*,c^*)}\geq \epsilon$ and $f_\G$ is continuous, we have that on $K$, $f_\G$ attains its minimum $\epsilon_0>0$, as desired.\footnote{To avoid calculations, here we use a basic fact from topology that any continuous function attains its minimum on a compact set (a standard result from topology). See e.g. \cite{munkresTopology2013}. Of course, one can also prove it by direct calculation.}
\end{proof}

Now, using the concept of bundle profiles, we can prove the necessity of \Cref{thm:main-result-almost-coprime}. We actually prove the following proposition.

\begin{proposition}\label{prop:2player-nonapproach}
Suppose that $\{\tau_i(n)\}_{n\geq 1}$ are the rationality levels of player $i$ ($i=1,2$) such that there exists a small enough $\delta>0$ satisfying
\[\limsup_{n\to\infty}\frac{\gcd(\tau_1(n),\tau_2(n))}{\min\{\tau_1(n),\tau_2(n)\}}\geq3\delta.\]
Consider the NE $\sigma_*=(r^*,c^*)$ with $r^*=c^*=(\delta,1-\delta)$ in game $\G_\delta$ of \eqref{eq:modified-matching-pennies-RC} and a sequence of its unfoldings $\G^\uf_\delta(n)$ with rationality levels $\tau_1(n)$ and $\tau_2(n)$, respectively. Then there exists some $\epsilon_0>0$ such that for any sequence of strategy profiles $s^\s{n}$ of $\G^\uf_\delta(n)$, we have
\[\limsup_{n\to\infty}f_{\G_\delta^\uf}\left(s^\s{n}\right)\geq\epsilon_0.\]
\end{proposition}

\begin{proof}

Let $\rho(n)=\gcd(\tau_1(n),\tau_2(n))$. By taking a subsequence, suppose without loss of generality that $\tau_1(n)\leq\tau_2(n)$ and there exists the limit
\[\lim_{n\to\infty}\frac{\rho(n)}{\tau_1(n)}\geq 3\delta.\]
Take a sequence of strategy profiles $s^\s{n}$ of $\G^\uf_\delta(n)$. Let $\mu_j^\s{n}$ ($j=1,\dots,\rho(n)$) be the bundle profiles of $s^\s{n}$.
    
First, we claim that for a large enough $n$, the minimum nonzero component of any bundle profile $\mu_{j,1}^\s{n}$ is greater than $2\delta$. Recall the definition of bundle folding:
\[\mu_{j,1}^\s{n}(a)=\frac{1}{\tau_1(n)/\rho(n)}\left|\left\{t\in[0,\tau_1(n)/\rho(n)-1]:m_1^\s{n}(j+t\cdot \rho(n))=a\right\}\right|.\]
Consider $a$ such that $\mu_{j,1}^\s{n}(a)>0$. If we write $\mu_{j,1}^\s{n}(a)=p_n/q_n$ for coprime $p_n$, $q_n$, we must have $q_n \le \tau_1(n)/\rho(n)$, so $\mu_{j,1}^\s{n}(a)\geq 1/q_n\geq\rho(n)/\tau_1(n)$. Since $\lim_{n\to\infty}\rho(n)/\tau_1(n)\geq 3\delta$, for sufficiently large $n$, $\rho(n)/\tau_1(n)>2\delta$, as desired.

A direct consequence of this claim is that for any bundle profile $j$, $\norm{\mu_j^\s{n}-\sigma_*}\geq\delta$. By \Cref{lemma:reverse-f_G-Lipschitz}, there exists some $\epsilon_0>0$ such that for any $j$,
\[f_{\G}\left(\mu_j^\s{n}\right)=\max_{i=1,2}\left\{f^i_{\G}\left(\mu_j^\s{n}\right)\right\}\geq 2\epsilon_0.\]
By the pigeonhole principle, for each player $i \in \{1,2\}$, let 
\[J_i(n) = \left\{j \in \{1, \dots, \rho(n)\} : f^i_\G(\mu_j^\s{n}) \ge 2\epsilon_0 \right\}.\] Then for at least one player $i$, we must have $|J_i(n)| \ge \rho(n)/2$. According to \Cref{lemma:no-coprime-cal-f-val},
\[
    f_{\G_\delta^\uf}\left(s^\s{n}\right)=\frac{1}{\rho(n)}\max_{i=1,2}\left\{\sum_{j=1}^{\rho(n)} f^i_{\G}\left(\mu_j^\s{n}\right)\right\}.
\]
Thus, we have
\[f_{\G_\delta^\uf}\left(s^\s{n}\right)\geq\frac{1}{\rho(n)}\cdot\frac{\rho(n)}{2}\cdot(2\epsilon_0) = \epsilon_0.\qedhere\]

\end{proof}
\end{document}